\documentclass[amsmath,secnumarabic,floatfix,amssymb,nofootinbib,nobibnotes,letterpaper,11pt,tightenlines,showkeys]{revtex4}

\usepackage{times}

\usepackage{geometry}
\usepackage{amssymb}
\usepackage{latexsym, amsmath, amscd,amsthm}
\usepackage{graphicx}
\usepackage[percent]{overpic}
\usepackage{units}
\usepackage{hyperref}
\PassOptionsToPackage{caption=false}{subfig}
\usepackage[lofdepth]{subfig}
\usepackage{algorithm}
\usepackage{algorithmicx}
\usepackage{algpseudocode}

\newtheorem{theorem}{Theorem}

\newtheorem{lemma}[theorem]{Lemma}
\newtheorem{proposition}[theorem]{Proposition}

\theoremstyle{definition}

\newcommand{\R}{\mathbb{R}}

\newcommand{\Pol}{\operatorname{Pol}}

\newcommand{\cross}{\times}

\newcommand{\Vol}{\operatorname{Vol}}

\newcommand{\dy}{\, \mathrm{d}y}
\newcommand{\dx}{\, \mathrm{d}x}
\newcommand{\sinc}{\operatorname{sinc}}

\newcommand{\SO}{\operatorname{SO}}

\newcommand{\PolHat}{\widehat{\Pol}(n)}

\let\mgp=\marginpar \marginparwidth18mm \marginparsep1mm
\def\marginpar#1{\mgp{\raggedright\tiny #1}}

\let\lbl=\label
\def\label#1{\lbl{#1}\ifinner\else\marginpar{\ref{#1} #1}\ignorespaces\fi}

\begin{document}
\title[]{A fast direct sampling algorithm for equilateral closed polygons}
\author{Jason Cantarella}
\altaffiliation{Department of Mathematics, University of Georgia, Athens GA}
\noaffiliation
\author{Bertrand Duplantier}
\altaffiliation{Institut de Physique Th\'eorique, CEA/Saclay, Gif-sur-Yvette Cedex, France}
\noaffiliation
\author{Clayton Shonkwiler}
\altaffiliation{Department of Mathematics, Colorado State University, Fort Collins CO}
\noaffiliation
\author{Erica Uehara}
\altaffiliation{Department of Physics, Ochanomizu University, Tokyo, Japan}
\noaffiliation

\begin{abstract}
Sampling equilateral closed polygons is of interest in the statistical study of ring polymers. Over the past 30 years, previous authors have proposed a variety of simple Markov chain algorithms (but have not been able to show that they converge to the correct probability distribution) and complicated direct samplers (which require extended-precision arithmetic to evaluate numerically unstable polynomials). We present a simple direct sampler which is fast and numerically stable, and analyze its runtime using a new formula for the volume of equilateral polygon space as a Dirichlet-type integral.
\end{abstract}

\keywords{closed random walk; random knot; random polygon; crankshaft algorithm, polygonal fold algorithm}

\date{\today}

\maketitle

\section{Introduction}
%
It is easy to sample random open polygons with equal edgelengths in $\R^3$: just pick steps independently and uniformly on the sphere. This model has been analyzed by authors since Lord Rayleigh~\cite{Rayleigh:1919do}. It is much harder to sample random \emph{closed} polygons, since the closure condition imposes subtle global correlations between edge directions. This problem is of interest in the statistical physics of polymers, as the closed equilateral polygon is a model for a ring polymer under ``$\theta$-conditions" (see the excellent survey~\cite{Orlandini:2007kn} for many applications of these kinds of models in physics and biology). A wide variety of sampling algorithms for random closed polygons have been proposed in the literature~\cite{Cantarella:2013wl, MR95g:57016,Vologodskii:1979ik,Klenin:1988dt,Anonymous:2010p2603,Moore:2004ds,Varela:2009cda}. Of these, only two have been shown rigorously to converge to the correct distribution on polygon space: the ``toric symplectic Markov chain Monte Carlo (TSMCMC)" algorithm of two of the present authors~\cite{Cantarella:2013wl} and the direct ``sinc integral method'' of Moore and Grosberg~\cite{Moore:2005fh}, discovered independently by Diao, Ernst, Montemayor, and Ziegler \cite{Diao:2011ie,Diao:wt,Diao:2012dza}. 

The purpose of this paper is to propose a direct sampling algorithm which improves on the sinc integral method. To describe it, we need to introduce a system of coordinates for polygon space. We start by fixing some notation. For any $n$-gon in $\R^3$, we let $v_1, \dots, v_n$ be the vertices of the polygon and $e_1, \dots, e_n$ be the edges (so $e_i = v_{i+1} - v_i$, where indices will always be interpreted cyclically). We will assume that $|e_i| = 1$, so our polygons are equilateral. The space of such $n$-gons with $v_1$ at the origin is a compact probability space $\Pol(n)$. The standard measure on $\Pol(n)$ is constructed by sampling $e_i$ uniformly and independently from the standard measure on the unit sphere, conditioned on the hypothesis that the polygon closes -- that is, that $\sum e_i = 0$. Since we are interested in the \emph{shapes} of polygons, we consider the moduli space $\PolHat = \Pol(n)/\SO(3)$ of polygons up to rotation; the measure on $\PolHat$ is simply the pushforward of the above conditional measure by the quotient map.

Joining all vertices of an $n$-gon to the first vertex, as in Figure~\ref{fig:assembly} (far left) creates a collection of $n-3$ triangles.  It is helpful to think of the $n$-gon as the boundary of the piecewise linear surface composed of these triangles. The embedding of the surface is determined by the lengths $d_i$ of the diagonals joining $v_1$ and $v_{i+2}$ and the dihedral angles $\theta_i$ between the triangles meeting at each diagonal. We can reconstruct the surface (and hence, the polygon) up to a rotation in space from this data as in Figure~\ref{fig:assembly}. This means that diagonal lengths and dihedral angles form a system of coordinates for $\PolHat$. They turn out to be ``action-angle'' coordinates for this space (in the sense of symplectic geometry, see~\cite{Cantarella:2013wl}). The angle coordinates are chosen independently, but the action coordinates (diagonal lengths) are not; since they are sidelengths of various triangles, they obey a system of triangle inequalities defining a convex polytope $\mathcal{P}_n \subset \R^{n-3}$.

\begin{figure}[t]
\begin{overpic}[width=1.3in]{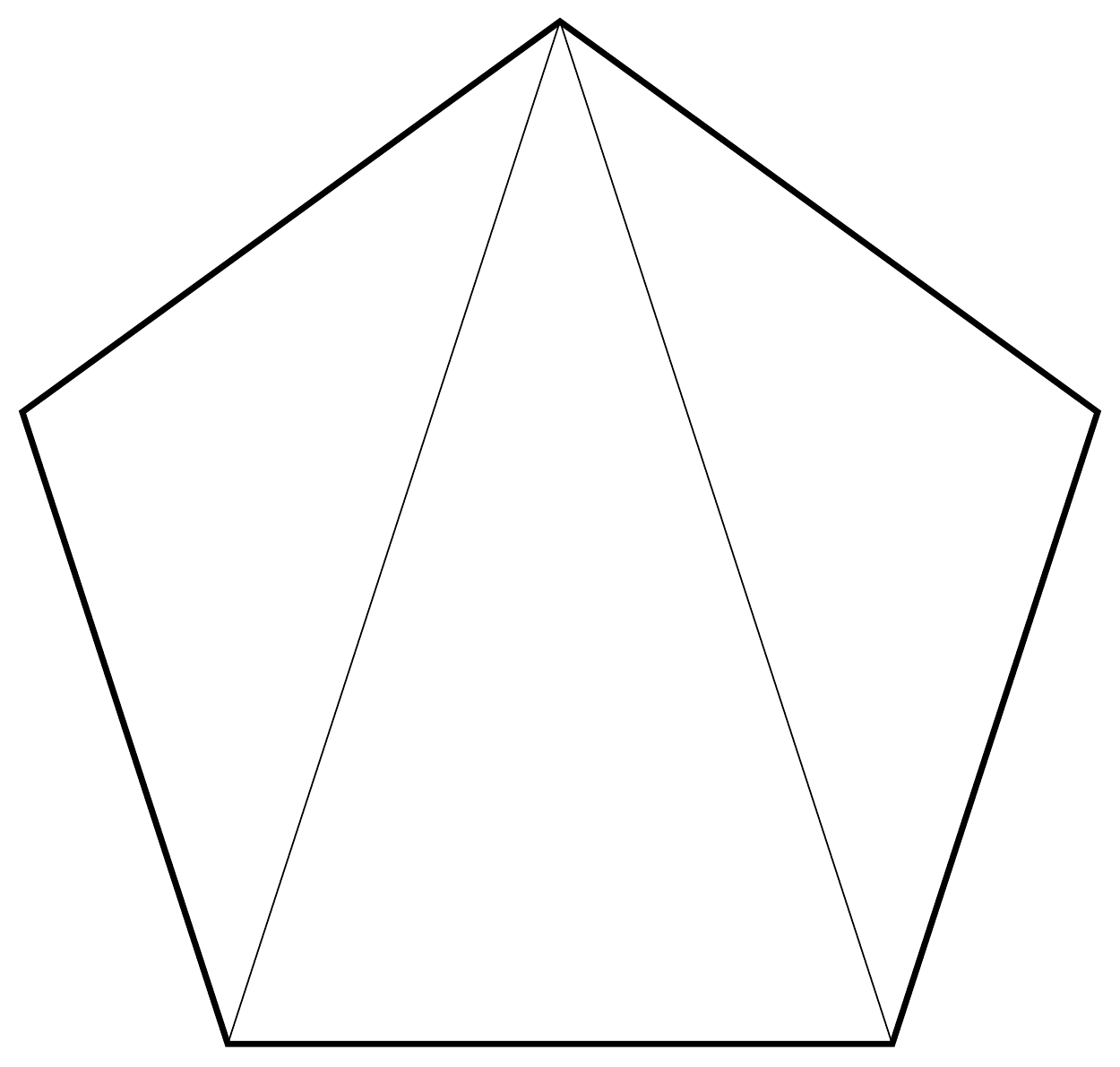}
\put(45,98){$v_1$}
\put(-9,57){$v_2$}
\put(17,-5){$v_3$}
\put(77,-5){$v_4$}
\put(100,57){$v_5$}
\end{overpic}
\hfill
\begin{overpic}[width=1.85in]{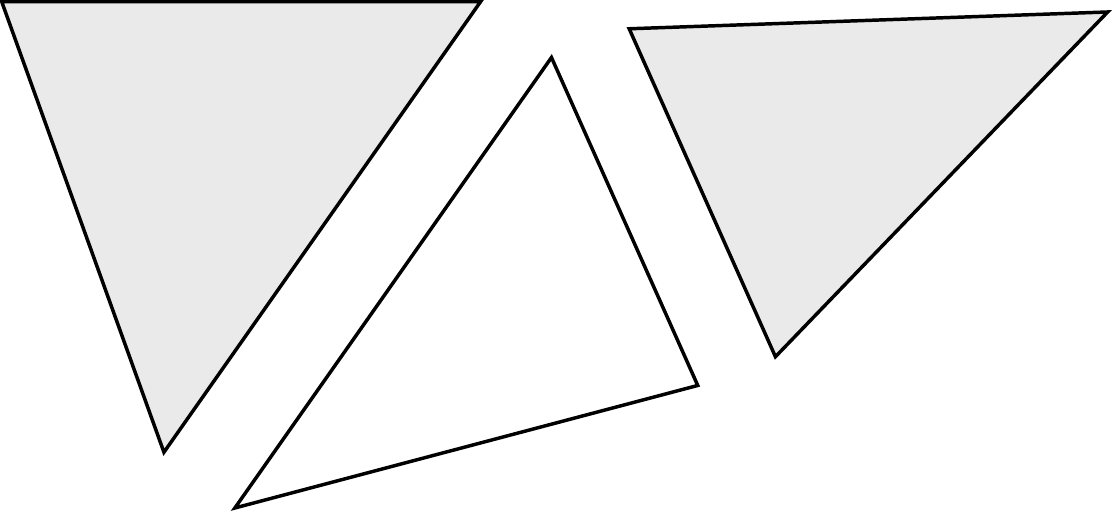}
\put(29.5,22.5){$d_1$}
\put(57.5,24){$d_2$}
\end{overpic}
\hfill
\begin{overpic}[width=1.1in]{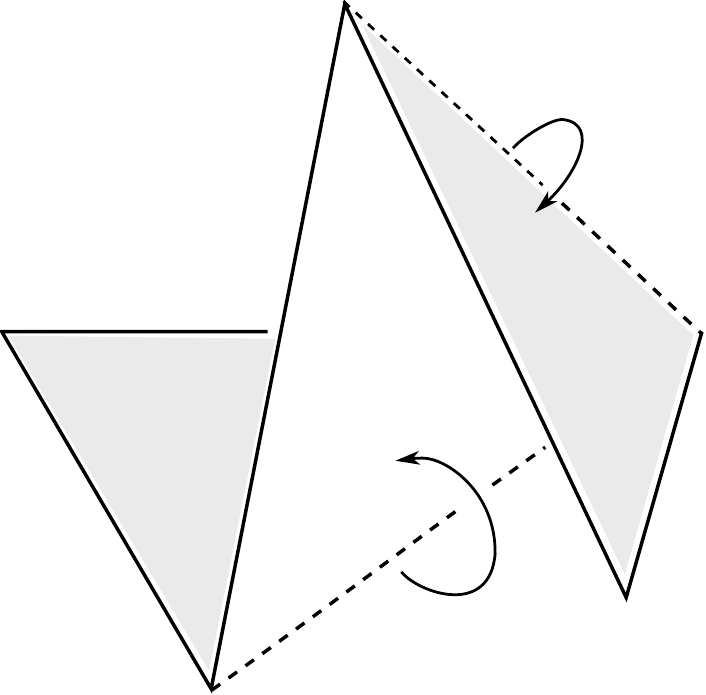}
\put(72,15){$\theta_1$}
\put(82,84){$\theta_2$}
\end{overpic}
\hfill
\begin{overpic}[width=1.1in]{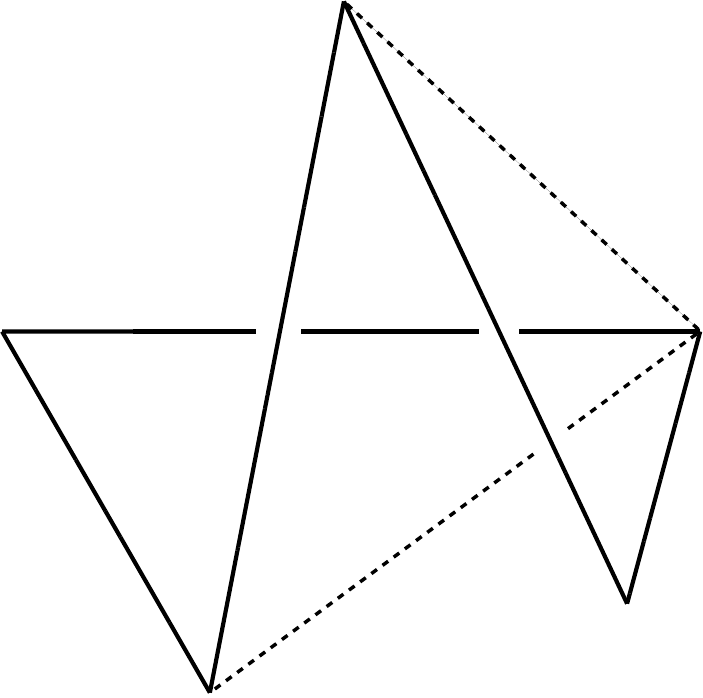}
\end{overpic}
\caption{This figure shows how to construct an equilateral pentagon from diagonals and dihedrals. On the far left, we see the fan triangulation of an (abstract) pentagon. Given diagonal lengths $d_1$ and $d_2$ of the pentagon which obey the triangle inequalities, we build the three triangles in the triangulation from their side lengths, as in the picture at middle left. Given dihedral angles $\theta_1$ and $\theta_2$, we can embed these triangles as a surface in space, as in the picture at middle right. The final space polygon, which is the boundary of this triangulated surface, appears far right.}
\label{fig:assembly} 
\end{figure}
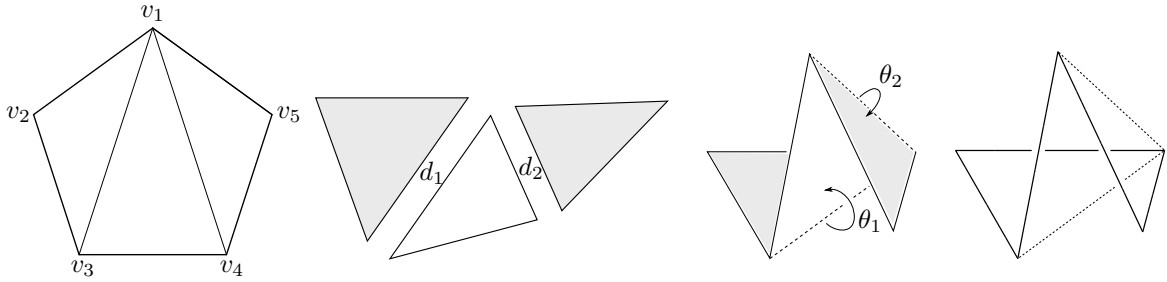

The sinc integral method constructs a polygon in action-angle coordinates stepwise by sampling each successive diagonal length $d_i$ from a piecewise polynomial pdf determined by the previous diagonal lengths and the number of diagonals remaining to be sampled (this pdf is determined by an integral of powers of the $\sinc$ function). The corresponding dihedral angle is sampled uniformly. This algorithm produces perfect samples, but is slow and difficult to implement. For instance, the text file giving the coefficients of the polynomials needed to sample a random closed 95-gon is over 25 megabytes in length. These polynomials have very large coefficients which almost cancel, requiring the use of extended precision arithmetic to evaluate stably\footnote{Hughes discusses these methods in Section 2.5.4 of his book on random walks~\cite{hughes1995random}, attributing the formula rederived by Moore and Grosberg~\cite{Moore:2005fh} to a 1946 paper of Treloar~\cite{TF9464200077}. The problems with evaluating these polynomials accurately were known by the 1970's, when Barakat~\cite{0301-0015-6-6-008} derived an alternate expression for this probability density based on Fourier transform methods.}.

Our improved method also constructs polygons in action-angle coordinates, but we construct all the diagonal lengths simultaneously. In previous work~\cite{Cantarella:2013wl}, two of us (Cantarella and Shonkwiler) used symplectic geometry to show that the sinc integral pdfs are basically marginals of the uniform distribution on the polytope $\mathcal{P}_n$. In fact, we show that it suffices to sample diagonal lengths and dihedral angles independently and uniformly from $\mathcal{P}_n \cross T^{n-3}$, where $T^{n-3} = (S^1)^{n-3}$ is the $(n-3)$-dimensional torus (see Theorem~\ref{prop:polygon sampling}), so the real problem is sampling $\mathcal{P}_n$.

In this paper, we show that a linear transformation of the polytope $\mathcal{P}_n$ is an unexpectedly large subset of the $(n-3)$-dimensional hypercube: the relative volume fraction is $\sim n^{-3/2}$ (Theorem~\ref{thm:volume bound}). Since samples of points in the hypercube take time $\sim n$ to construct, this means that we can use rejection sampling to find points in the polytope in (expected) $\Theta(n^{5/2})$ time. We call our method the \emph{Action-Angle Method}.

A test C implementation is quite efficient, generating random 100 edge polygons at 2400 samples per second, 1000 edge polygons at 30 samples per second, and 3000 edge polygons at 2 samples per second on a laptop. We redo some previous experiments of Alvarado et al.~\cite{Anonymous:2010p2603} to exercise our implementation and make the curious observation that the log-log plot of the rank statistic for knot types seems remarkably linear.



\section{Action-angle coordinates and the polytope~$\mathcal{P}_n$}
%
We start by describing the polytope $\mathcal{P}_n$ of diagonal lengths more explicitly. Each $d_i$ is a side length of two triangles, and the corresponding triangle inequalities are 
\begin{equation}
0 \leq d_1 \leq 2 
\qquad 
\begin{matrix} 
1 \leq d_i + d_{i+1} \\
-1 \leq d_i - d_{i+1} \leq 1 
\end{matrix}
\qquad
0 \leq d_{n-3} \leq 2 
\label{eq:fan polytope}
\end{equation} 
If $\mathcal{P}_n \subset \R^{n-3}$ is the polytope defined by these inequalities, the action-angle coordinates are defined on $\mathcal{P}_n \cross T^{n-3}$. In \cite{Cantarella:2013wl}, two of us (Cantarella and Shonkwiler) proved that

\begin{theorem}\label{prop:polygon sampling}
If we take the uniform measure on $\mathcal{P}_n \cross T^{n-3}$, the reconstruction map ${\alpha: \mathcal{P}_n \cross T^{n-3} \rightarrow \PolHat}$ defining action-angle coordinates is measure-preserving.
\end{theorem}

Put another way, just as one can sample the sphere uniformly by choosing the cylindrical coordinates $z$ uniformly on $[-1,1]$ and $\theta$ uniformly on $[0,2\pi)$, one can sample polygon space uniformly by choosing $\vec{d}$ uniformly from $\mathcal{P}_n$ and $\vec{\theta}$ uniformly from $T^{n-3}$. It is not important to understand the proof of the above theorem for the rest of this paper. However, we can summarize by saying that previous authors had shown that equilateral polygon space is (almost) a special kind of manifold called a toric symplectic manifold~\cite{Kapovich:1996p2605}. Further, the coordinates above have a special relationship to that structure: rotation around a diagonal is a symmetry of the space (parametrized by the dihedral angle) and the length of the diagonal is the corresponding conserved quantity (or action). The Duistermaat-Heckman theorem~\cite{Duistermaat:1982hq} then implies that the pushforward of the symplectic volume of the space to the ``moment polytope'' $\mathcal{P}_n$ must be the uniform measure on the polytope. 
All this is explained in some detail in \cite{Cantarella:2013wl}. The results are more general than we are using here -- for instance, we could fix the edgelengths to be different numbers or triangulate the polygon using a different collection of diagonals and the theory would still work.


\section{A change of coordinates and a volume estimate}
%
We must now sample $\mathcal{P}_n$. We first make a linear transformation of the variables. We extend the definition of the $d_i$ to include $d_0 = |v_2 - v_1| = 1$ and $d_{n-2} = |v_n - v_1| = 1$. Looking at the system of inequalities~\eqref{eq:fan polytope}, it is natural to rewrite this system in terms of new variables $s_1, \dots, s_{n-2}$ where $s_i = d_{i} - d_{i-1}$. Since $\sum s_i = d_{n-2} - d_0 = 0$, the last variable $s_{n-2}$ is determined uniquely by the previous $n-3$ variables. Translating the system~\eqref{eq:fan polytope} from $d$ to $s$ coordinates yields the system of inequalities
\begin{equation}
	    \underbrace{-1 \leq s_i \leq 1}_{|d_i - d_{i+1}| \leq 1}, \quad  -1 \leq \sum_{i=1}^{n-3} s_i \leq 1  \quad \text{ and }\quad \underbrace{\sum_{j=1}^{i} s_j + \sum_{j=1}^{i+1} s_{j} \geq -1}_{d_i + d_{i+1} \geq 1}.
\label{eq:chordlength excursions}
\end{equation}

\begin{proposition}\label{prop:moment polytope volume}
	The $n-3$-dimensional moment polytope $\mathcal{P}_n$ is the image under a volume-preserving linear transformation of the $(n-3)$-dimensional polytope $\mathcal{C}_n \subset [-1,1]^{n-3}$ defined by the inequalities~\eqref{eq:chordlength excursions}. Moreover, the volume of $\mathcal{C}_n$ (and $\mathcal{P}_n$) is
\begin{equation} 
-\frac{1}{2 (n-3)!}\sum
   _{k=0}^{\left\lfloor
   \nicefrac{n}{2}\right\rfloor }
   (-1)^k \binom{n}{k} (n - 2k)^{n-3}
 = \frac{2^{n-1}}{2\pi}\int_{-\infty}^\infty \frac{\sin^n x}{x^{n-2}} \dx.
\label{eq:volume}
\end{equation}
\end{proposition}

\begin{proof}
	The linear transformation from $\mathcal{C}_n$ to $\mathcal{P}_n$ changes from $s$ back to $d$ coordinates; explicitly, $d_j = 1+\sum_{i=1}^j s_i$. The domain of this mapping is the subset of~$\R^{n-3}$ where $-1 \leq \sum s_i \leq 1$, and the range is the affine subspace of $\R^{n-1}$ where $d_0 = d_{n-2} = 1$. A computation reveals that this map preserves $(n-3)$-dimensional volume, so $\Vol \mathcal{C}_n = \Vol \mathcal{P}_n$.
	
	But the volume of $\mathcal{P}_n$ is known! By results of \cite{Takakura:2002tz,Khoi:2005ch,Mandini:2008wq}, the volume of the entire polygon space $\PolHat$ is exactly
\begin{equation}\label{eq:Mandini}
\Vol \PolHat = -\frac{(2\pi)^{n-3}}{2(n-3)!} \sum_{k=0}^{\lfloor \nicefrac{n}{2}\rfloor} (-1)^{k} \binom{n}{k} (n-2k)^{n-3}.
\end{equation}
This volume is the product of the volume $(2\pi)^{n-3}$ of the $(n-3)$-dimensional torus and the volume of $\mathcal{P}_n$. The result follows after dividing \eqref{eq:Mandini} by $(2\pi)^{n-3}$ and using the general Dirichlet-type integral formula:
\begin{equation}\label{eq:dirichlet integral}
  \frac{(-1)^p}{(n-2p-1)!} \sum_{k=0}^{\left\lfloor\nicefrac{n}{2}\right\rfloor} (-1)^k \binom{n}{k} (n-2k)^{n-2p-1} = \frac{2^n}{2\pi}\int_{-\infty}^\infty \frac{\sin^n x}{x^{n-2p}} \dx.
\end{equation}
This equation goes back at least to Edwards~\cite[p.~212]{edwards1922treatise}, who attributes it to Wolstenholme (cf.~\cite[p.~1703]{Borwein:2001bw,Amdeberhan:2007wo,Nathan:2014dn,Weisstein:2003vq}), but we also give a short proof in the appendix.
\end{proof}

We note that the Dirichlet integral formulation above of the volume of the moment polytope appears to be new. 

We now know that sampling $\mathcal{P}_n$ is equivalent to sampling $\mathcal{C}_n$. This leads to our main result, which implies that we can sample $\mathcal{C}_n$ efficiently by rejection sampling the hypercube $[-1,1]^{n-3}$:

\begin{theorem}\label{thm:volume bound}
Let $p_n$ be the probability that a random point in $[-1,1]^{n-3}$ is in $\mathcal{C}_n$. For large $n$,  
\[
	p_n \sim \frac{6 \sqrt{6}}{\sqrt{\pi}}\frac{1}{n^{\nicefrac{3}{2}}}
\]
\end{theorem}

\begin{proof}
Making the substitution $x = \nicefrac{y}{\sqrt{n}}$ in~\eqref{eq:volume}, we get
\[
\Vol \mathcal{C}_{n} = \frac{2^{n-1}}{2\pi}\int_{-\infty}^\infty \left( \frac{ \sin\left(\nicefrac{y}{\sqrt{n}}\right)}{\nicefrac{y}{\sqrt{n}}}\right)^n \frac{y^{2}\dy}{n^{\nicefrac{3}{2}}} = \frac{2^{n-1}}{2\pi}\int_{-\infty}^\infty  \sinc^n \left(\nicefrac{y}{\sqrt{n}} \right)  \frac{y^{2}\dy}{n^{\nicefrac{3}{2}}}.
	\]
	Now the Taylor expansion of $\sinc \frac{y}{\sqrt{n}}$ is $1 - \nicefrac{y^2}{6n} +  o(\nicefrac{1}{n})$. But $\lim_{n\rightarrow\infty} \left(1 - \nicefrac{a}{n} + o(\nicefrac{1}{n}) \right)^n = e^{-a}$, so the integral above is asymptotically
	\[
		\frac{2^{n-1}}{2\pi}\frac{1}{n^{\nicefrac{3}{2}}} \int_{-\infty}^\infty e^{-\nicefrac{y^2}{6}} y^{2} \dy = 3 \, \sqrt{\frac{3}{\pi }}
   \, 2^{n-\frac{3}{2}} \, \frac{1}{n^{\nicefrac{3}{2}}},
	\]
	and dividing by the volume of the hypercube gives the result.
\end{proof}

\section{The Action-Angle Method}

We can now state our algorithm for sampling equilateral polygons very simply:
\begin{algorithmic}
\Procedure{ActionAngleSample}{$n$}\Comment{Generate closed equilateral $n$-gon}
\Repeat
	\Repeat
		\State Sample $n-3$ i.i.d.\ step lengths $(s_1,\dots,s_{n-3})$ uniformly from $[-1,1]$ 
	\Until $-1 \leq \sum_{j=1}^{n-3} s_j \leq 1$.
	\State Let $s_{n-2} = -\sum_{j=1}^{n-3} s_j$.
	\State Construct diagonals $d_i = 1 + \sum_{j=1}^i s'_j$, noting $d_0 = 1$, $d_{n-2} = 1$.
\Until $d_{i} + d_{i+1} \geq 1$ for all $i$.
\State Sample $n-3$ i.i.d.\ dihedral angles $\theta_i$ uniformly from $[0,2\pi]$.
\State Reconstruct $P$ from diagonals $d_1, \dots, d_{n-3}$ and dihedrals $\theta_1, \dots, \theta_{n-3}$.
\EndProcedure
\end{algorithmic}

\begin{theorem}
The Action-Angle Method generates uniform random samples of closed, equilateral $n$-gons in expected time $\Theta(n^{5/2})$. 
\end{theorem}

\begin{proof}
The fact that the algorithm generates samples according to the correct probability distribution is now easy to check. Theorem~\ref{prop:polygon sampling} tells us that we need only sample the diagonals from the moment polytope $\mathcal{P}_n$ uniformly. Since a linear transformation of the uniform distribution on $\mathcal{P}_n$ is uniform, it suffices to sample $\mathcal{C}_n$ uniformly. We do this by rejection sampling in the loops above. 

By Theorem~\ref{thm:volume bound} we have an acceptance ratio $\sim \frac{6\sqrt{6}}{\sqrt{\pi}}\frac{1}{n^{3/2}}$. Since sampling $n-3$ variates in the innermost loop requires linear time, the total time to produce a sample is $\Theta(n^{5/2})$. The postprocessing steps of generating dihedral angles and reconstructing the polygon are both linear in $n$, so they do not affect the time bound.\end{proof}

\section{Testing the algorithm}
%

We implemented the algorithm in C. Our implementation is incorporated into the freely available \texttt{plCurve} package of Ashton, Cantarella, and Chapman~\cite{plcurve}, as \verb|plc_random_equilateral_closed_polygon|.

To test it, we generated ensembles of polygons and computed sample averages for chordlengths, radius of gyration, and total curvature to compare against theoretical results. All of the sample means were within expected error of the theoretical values. A spot check which can be performed in a few seconds is to take 60,000 random 31- and 32-gons and compute mean total curvature, comparing to the theoretical values  $49.912$ and $51.482$ (rounded). We found $95\%$ confidence intervals of $49.902 \pm 0.0303$ and $51.475 \pm 0.0310$, respectively. 

When comparing various Markov chain algorithms, Alvarado et.\ al.~\cite{Anonymous:2010p2603} tested the number of distinct HOMFLY polynomials (the HOMFLY polynomial is a measure of knot type) observed when taking 10 million samples of equilateral 60-gons. We performed the same test with our direct sampling algorithm to see how the various Markov chain algorithms compared. This ran for about 5 CPU-hours on a laptop. The results are shown in Figure~\ref{fig: datacomp}.

\begin{figure}
\hphantom{.}
\hfill
\raisebox{\height}{
\begin{minipage}{2.4in}
\begin{ruledtabular}
\begin{tabular}{ll}
Algorithm        & distinct HOMFLYs \\ \hline
Polygonal Folds\footnotemark  & 2219 \\
Crankshaft Moves & 6110 \\
Hedgehog Method  & 1111 \\
Triangle Method  & 3505 \\ \hline
Action-Angle Method & $\geq$ 6371 
\end{tabular}
\end{ruledtabular}
\footnotetext[1]{100 million samples, instead of 10 million.}
\end{minipage}
}
\hfill
\includegraphics[width=3in]{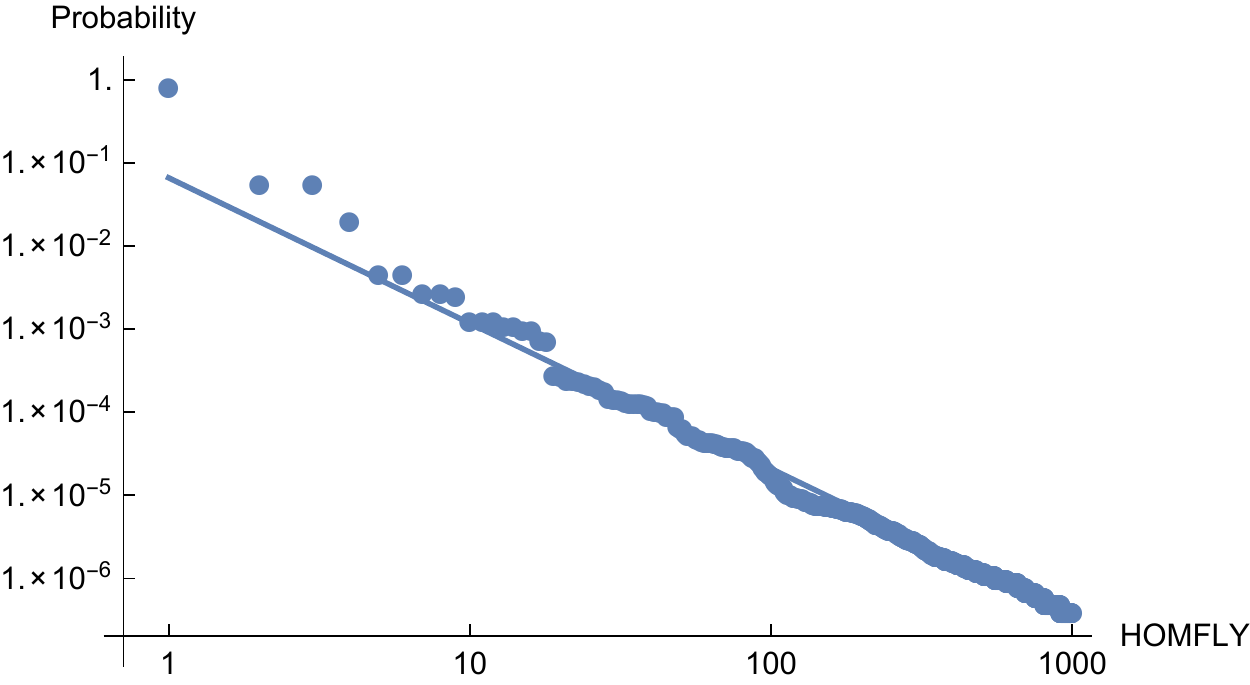}
\hfill

\caption{A comparison of various Markov chain sampling methods to our direct sampling method. The data in the table at left is from~\cite{Anonymous:2010p2603}; they computed the number of different HOMFLY polynomials observed over a sample of 10 million random 60-gons (100 million for the polygonal fold algorithm) generated by Markov chains starting at the regular $n$-gon. We see that these Markov chains do not seem to have converged to the standard probability measure, since direct sampling observes more topological types of polygons. The Action-Angle result is a lower bound since there were 42 polygons generated which were numerically singular. We did not compute a HOMFLY for these. The sizes of the regions of polygon space corresponding to different HOMFLY polynomials are distributed very unevenly. The right hand graph shows a log-log plot of the probability of the $n$-th most frequent HOMFLY polynomial in the sample; the straight line is the graph of $e^{-e} n^{-7/4}$.}
\label{fig: datacomp}
\end{figure}

\section{Future Directions}

In~\cite{Cantarella:2013wl}, two of us gave a modified version of the TSMCMC algorithm for sampling closed equilateral polygons in (rooted) spherical confinement -- where all vertices of the polygon are contained in a ball around the first vertex. The diagonal lengths of these polygons are sampled from a subpolytope of the polytope $\mathcal{P}_n$ of diagonals described above. It would be interesting to see if we could describe this subpolytope well enough to make rejection sampling realistic for these polygons as well.

Another interesting direction to pursue might be the role of the triangulation of the polygon by chords in defining the polytope to sample. Above, we joined each vertex to the first, and let the lengths of these diagonals define the sample polytope $\mathcal{P}_n$. However, the number of possible triangulations of the $n$-gon is given by the $(n-2)$-nd Catalan number, and every one of these triangulations defines a different polytope, each equally valid for sampling. Perhaps another one of these polytopes could be surrounded more tightly by a standard polytope (or even decomposed into simplices explicitly!), improving the efficiency of the method above.

Of course, the most important question in applying these methods to polymer physics is whether they can be extended to deal with spaces of polygons which are even more restricted, such as the space of polygons with excluded volume. But we do not yet understand the geometry of these polygon spaces well enough to fit them into our theory.

\section*{Acknowledgments}

We are grateful to Stephen DeSalvo, whose paper~\cite{DeSalvo:2014wa} and comments substantially influenced the final form of the Action-Angle Method, as well as to Ken Millett, Tetsuo Deguchi, Yuanan Diao, Claus Ernst, and Uta Ziegler for many discussions about the polygon sampling problem. This work was supported by grants from the Simons Foundation (\#354225, Clayton Shonkwiler, and \#284066, Jason Cantarella). We want to thank the Simons Center for Geometry and Physics for hosting the workshop on ``Symplectic and Algebraic Geometry in the Statistical Physics of Polymers'', attended by the authors, where some of the work for this paper was completed.

\appendix
\section{Proof of the Sum Formula for Dirichlet-Type Integrals} \label{appendix}

\begin{lemma}\label{lem:dirichlet integral formula}
	For nonnegative integers $p$ and $n$ with $n \geq 2p+1$,
	\[
		I_{n,p} := \frac{2^n}{2\pi}\int_{-\infty}^\infty \frac{\sin^n x}{x^{n-2p}} \dx = \frac{(-1)^p}{(n-2p-1)!} \sum_{k=0}^{\left\lfloor\nicefrac{n}{2}\right\rfloor} (-1)^k \binom{n}{k} (n-2k)^{n-2p-1}.
	\]
\end{lemma}

\begin{proof}
	Integrating by parts $n-2p-1$ times yields
	\[
		I_{n,p} = \frac{2^n}{2\pi(n-2p-1)!} \int_{-\infty}^\infty \frac{d^{n-2p-1}}{dx^{n-2p-1}} \left( \sin^n x \right) \frac{\dx}{x}.
	\]
	Expanding $\sin^n x=\left(\frac{e^{i x} - e^{-i x}}{2i}\right)^n$ with the binomial theorem, differentiating $n-2p-1$ times, and grouping the $k$ and $n-k$ terms produces
	\[
		I_{n,p} = \frac{(-i)^{2p}}{2\pi(n-2p-1)!} \sum_{k=0}^{\left\lfloor \nicefrac{n}{2}\right\rfloor} (-1)^k \binom{n}{k} (n-2k)^{n-2p-1} \int_{-\infty}^\infty \frac{e^{i(n-2k)x}-e^{-i(n-2k)x}}{i x} \dx.
	\]
	Since the definite integral is just $ \int_{-\infty}^\infty \frac{2\sin((n-2k)x)}{x}\dx = 2 \pi $, the result follows.
\end{proof}

\bibliography{excursions_extra,excursions_papers}

\end{document}